\documentclass[a4paper]{article}

\usepackage{amsmath,amsfonts,amssymb,amsthm,color,hyperref}

\newtheorem{defi}{Definition}
\newtheorem{teo}{Theorem}

\newtheorem{prop}[teo]{Proposition}
\newtheorem{cor}[teo]{Corollary}
\newtheorem{lem}[teo]{Lemma}
\newtheorem{rmk}{Remark }
\newtheorem{exm}{Example}
\newenvironment{proof1}[1]
{\textit{Proof of #1.}}{\hfill$\square$}

\renewcommand{\aa}{a_{0}}
\newcommand{\bb}{a_1}
\newcommand{\ee}{a_2}
\newcommand{\dok}{c_N}

\begin{document}
	
	
\title{Extended Hamiltonians and shift, ladder functions and operators}
\author{Claudia Maria Chanu, Giovanni Rastelli \\  Dipartimento di Matematica, \\ Universit\`a di Torino.  Torino, via Carlo Alberto 10, Italia.\\ \\ e-mail: claudiamaria.chanu@unito.it \\  giovanni.rastelli@unito.it }

\maketitle

\begin{abstract} In recent years, many natural Hamiltonian systems, classical and quantum, with constants of motion of high degree, or symmetry operators of high order, have been found and studied. Most of these Hamiltonians, in the classical case, can be included in the family of extended Hamiltonians, geometrically characterized by the structure of warped manifold of their configuration manifold. For the extended manifolds, the characteristic constants of motion of high degree are polynomial in the momenta of determined form. We consider here a different form of the constants of motion, based on the factorization procedure developed by S. Kuru, J. Negro and others. We show that an important subclass of the extended Hamiltonians admits factorized constants of motion and we determine their expression. The
classical constants may be non-polynomial in the momenta, but the factorization procedure allows, in a type of extended Hamiltonians, their quantization via shift and ladder operators, for systems of any finite dimension. 
\end{abstract}	
	
\section{Introduction}

In the past few years, the systematic study of classical Hamiltonian systems admitting polynomial constants of motion of high degree, and quantum Hamiltonian operators with symmetry operators of high order, produced several articles \cite{CKNd,Bo, KKM,KKMncf,6bis,7bis,Qu,7ter,Ra2,TTW}, many of them focused on the superintegrability of the systems considered. In several articles \cite{CDR0,CDRPol,CDRfi,CDRgen,CDRsuext,CDRraz,TTWcdr, CDRpw}, we, together with Luca Degiovanni, were able to unify most of those  Hamiltonian systems under the common structure of \emph{extended Hamiltonians}. For the case of natural Hamiltonians, characterized geometrically in \cite{CDRgen}, to be an extended Hamiltonian implies that the configuration manifold is a particular \emph{warped Riemannian manifold} \cite{Ta}, equipped with a conformal  Killing vector field of special type.  Extended Hamiltonians allow the existence of privileged coordinate systems adapted to the warped geometry, with the property of splitting a given $(N+1)$-dimensional  system into  one-dimensional and  $N$-dimensional parts. 	The most important consequence of the extended Hamiltonian structure, in the classical case, is the existence of constants of motion, that we call \emph{characteristic}, determined by the power of an operator applied to a well determined function, in correspondence of a rational parameter employed in the construction of the extended Hamiltonian itself. The classical constants of motion of high degree in the Tremblay-Turbiner-Winternitz (TTW) system (and in its particular cases: three-particles Jacobi-Calogero and Wolfes systems), in the Post-Winternitz  system, and in their  multi-dimensional and non-zero-curvature manifold generalizations, as shown in \cite{TTWcdr}, are exactly of that type. Until now, we did not have a similar unified procedure for determining the symmetry operators of the quantum Hamiltonian corresponding to an extended Hamiltonian. In the cases considered in literature, essentially the TTW system and its generalizations, the symmetry operators are determined case by case, making use of the complete separation of variables of the Hamiltonians. The symmetry operators are therefore always obtained by some factorization procedure. One of these factorization techniques is that developed by S. Kuru, J. Negro and  co-authors \cite{KN0, B, B1, CKNd}. The procedure is based on the well known factorization in creation-annihilation, or shift and ladder, operators, developed firstly for the harmonic oscillator. This technique is applied to the classical case also, obtaining the factorization of classical constants of motion in shift and ladder functions. As S. Kuru and J. Negro pointed out in a private communication, the basic ansatz of their factorization technique (the existence of particular ladder functions) implies the structure of an extended manifold for the configuration manifold of the system. By assuming this ansatz, we are able to show here how the classical characteristic constants of motion of  extended Hamiltonian systems  factorize in shift and ladder functions. Moreover, we show that, in some important cases, the extended Hamiltonian structure implies the validity of the ansatz, that  it is then no longer a restriction.  Most importantly, in analogy with these results we are able to build a general quantization procedure of the characteristic constant of motion,  for one of the three possible forms of extensions, in any finite dimension, the model being the shift-ladder factorization of above. The content of the article, Section by Section, is the following. 
In Section \ref{s2} we recall the basics of the theory of extended Hamiltonians; in Section \ref{s3} we show that the existence of certain ladder functions implies the factorization of the characteristic first integrals of the extended Hamiltonians
(Theorem \ref{GN}); in Section \ref{s4}, on the contrary, we show that a large class of extended Hamiltonians admit a factorized first integral, according to Theorem \ref{GN}. In Section \ref{s5}, we make a comparison between the Kuru-Negro factorization and the one proposed here. In Section \ref{s6}, we consider the quantum version of the procedure and we propose a quantized extension procedure, holding for a subclass of the extended Hamiltonians, based on the Kuru-Negro technique. Given a type of extended Hamiltonian, we obtain an extended quantum Hamiltonian and its factorized characteristic symmetry operator (Theorem \ref{teoq}), the construction of which is described in subsections \ref{6.1}-\ref{6.3}. 
	
\section{Extensions: definitions and relevant properties}\label{s2}
Let $L(q^i,p_i)$ be a Hamiltonian with $N$ degrees of freedom, that is defined on the cotangent bundle of an $N$-dimensional manifold.
We say that $L$ {\em admits extensions}, if there exists $(c,c_0)\in \mathbb R^2
- \{(0,0)\}$ such that there exists a non null solution $G(q^i,p_i)$ of
\begin{equation}\label{e1}
X_L^2(G)=-2(cL+c_0)G,
\end{equation}
where $X_L$ is the Hamiltonian vector field of $L$.
If $L$ admits extensions, then, for any $\gamma(u)$ solution of the ODE
\begin{equation}\label{eqgam}
\gamma'+c\gamma^2+C=0,
\end{equation}
depending on the arbitrary constant parameter 
$C$,	
we say that any Hamiltonian 
$H(u,q^i,p_u,p_i)$ with $N+1$ degree of
freedom of the form
\begin{equation}\label{Hest}
H=\frac{1}{2} p_u^2-k^2\gamma'L+ k^2c_0\gamma^2
+\frac{\Omega}{\gamma^2},
\qquad k\in \mathbb{Q}-\{0\}, \   \Omega\in\mathbb{R}
\end{equation}
is an {\em  extension of $L$}.

Extensions of Hamiltonians where introduced
in \cite{CDRfi} and studied because they admit polynomial in the momenta first integrals generated via a recursive algorithm. Moreover, the degree of the first integrals is related with the (choice of) numerator and denominator of the rational parameter $k=m/n$, $m,n\in \mathbb{N}^*$.
Indeed, for any $m,n\in \mathbb N-\{0\}$, 
let us consider the operator
\begin{equation}\label{Umn}
U_{m,n}=p_u+\frac m{n^2}\gamma X_L
\end{equation}

\begin{prop}\cite{CDRraz}
For $\Omega=0$,
the Hamiltonian (\ref{Hest}) 
is in involution with the function
	\begin{equation}\label{mn_int}
	K_{m,n}=U_{m,n}^m(G_n)=\left(p_u+\frac{m}{n^2} \gamma(u)   X_L\right)^m(G_n)
	\end{equation}
	where $m/n=k$ and $G_n$ is the $n$-th term of the recursion
	\begin{equation}\label{rec}
	G_1=G, \qquad G_{n+1}=X_L(G)\,G_n+\frac{1}{n}G\,X_L(G_n), 
	\end{equation}
starting from any solution $G$ of (\ref{e1}).
\end{prop}

\begin{rmk} \rm The  functions $G_n$ satisfy
	$$
	X_L^2G_n=-2n^2 (cL+c_0)G_n.
	$$
\end{rmk}

For $\Omega\neq 0$, the recursive construction of a first integral is more complicated:  we consider the following function, depending on two non-zero integers
$2s,r$ (the first one has to be even)
	\begin{equation}
	\bar K_{2s,r}=\left(U_{2s,r}^2+2\Omega \gamma^{-2}\right)^s(G_r).\label{ee2}
	\end{equation}
where 
the operator $U^2_{2s,r}$ is defined {according to
	(\ref{Umn})} as
$$
U^2_{2s,r}=\left( p_u+\frac{2s}{r^2} \gamma(u)   X_L
\right)^2,
$$
and 
$G_r$ is, as in (\ref{mn_int}),  the $r$-th term of the recursion (\ref{rec}).  
with $G_1=G$ solution of (\ref{e1}).
For $\Omega=0$ the functions (\ref{ee2})
reduces to (\ref{mn_int}) and thus can be computed 
also when the first of the indices is odd. 

We have 

\begin{teo}\cite{TTWcdr}
For any $\Omega\in \mathbb{R}$,
the Hamiltonian (\ref{Hest}) with $k=m/n$ 
satisfies
	for  $m=2s$, 
		\begin{equation}\label{cp}
		\{H,\bar K_{m,n}\}=0,
		\end{equation}
for $m=2s+1$,
		\begin{equation}
		\label{cd}
		\{H ,\bar K_{2m,2n}\}=0.
		\end{equation}
\end{teo}

We call $K$ and $\bar{K}$, of the form (\ref{mn_int}) and
(\ref{ee2}) respectively, \emph{ characteristic first integrals}  of the corresponding extensions

\begin{rmk}
\rm
In \cite{CDRgen} it is proven that the ODE (\ref{eqgam}) defining $\gamma$ is a necessary condition in order to get a characteristic first integral of the form (\ref{mn_int}) or
(\ref{ee2}). 
According to the value of $c$, the explicit form of $\gamma(u)$ is given (up to constant translations of $u$) by
\begin{equation}\label{fgam}
\gamma= \left\{ 
\begin{array}{lc}
-C u & c=0 \\
\frac{1}{T_\kappa (cu)}=\frac{C_\kappa (cu) }{S_\kappa (cu)}
& c\neq 0
\end{array}
\right.
\end{equation}
where $\kappa=C/c$ is the ratio of the constant parameters appearing in (\ref{eqgam}) and $T_\kappa$, $S_\kappa$ and $C_\kappa$  are the trigonometric tagged functions
 \cite{7ter}
		$$
		S_\kappa(x)=\left\{\begin{array}{ll}
		\frac{\sin\sqrt{\kappa}x}{\sqrt{\kappa}} & \kappa>0 \\
		x & \kappa=0 \\
		\frac{\sinh\sqrt{|\kappa|}x}{\sqrt{|\kappa|}} & \kappa<0
		\end{array}\right.
		\qquad
		C_\kappa(x)=\left\{\begin{array}{ll}
		\cos\sqrt{\kappa}x & \kappa>0 \\
		1 & \kappa=0 \\
		\cosh\sqrt{|\kappa|}x & \kappa<0
		\end{array}\right.,
		$$
		$$
		T_\kappa(x)=\frac {S_\kappa(x)}{C_\kappa(x)},
		$$
(see also \cite{CDRraz} for a summary of their properties).
\end{rmk}	

It is proved in \cite{CDRfi,TTWcdr} that the  characteristic first integrals $K$ or $\bar K$  are functionally independent from $H$, $L$, and from any 
first integral $I(p_i,q^i)$ of $L$.
This means that the extensions of (maximally) superintegrable Hamiltonians	are 
(maximally) superintegrable Hamiltonians with one additional degree of freedom.

\begin{rmk} \rm
The equation (\ref{e1}) is the fundamental condition for the  extension procedure. When $L$ is a $N$-dimensional natural Hamiltonian with metric tensor $\mathbf g=(g_{ij})$ and $G$ does not depend on the momenta, (\ref{e1}) splits into several equations, the coefficients of monomials in the momenta, whose leading term is the Hessian equation
\begin{equation}\label{hesseq}
\nabla_i\nabla_j G+cGg_{ij}=0,\quad i,j = 1,\ldots ,N.
\end{equation}
This equation is of paramount importance in theory of warped manifolds, see for example \cite{Ta},
where it is 
proved, for manifolds of positive definite signature, the strict connection between the constant $c$ and the curvature of the metric $g$ (Theorem 2, where it is also proved that, if the metric is geodesically complete and $c\neq0$, then $g$ is of constant curvature $c$) and the warped structure of $g$ whenever the equation (\ref{hesseq}) admits solutions $G(q^k)$ (Lemma 1.2 of \cite{Ta}).

 However, by relaxing the hypothesis of the
 these statements (for example allowing the non geodesic completeness of the manifold in the case of $c\neq 0$), solutions of (\ref{hesseq}) can be found also in non-constant curvature manifolds, \cite{CDRmlbq}.
Moreover, the structure of the extensions
$H$ of natural Hamiltonians shows that the extended metric is again a warped metric.
\end{rmk}

We want to study other properties of
extensible and extended natural Hamiltonians, related 
with the existence of ladder and shift functions as well as ladder and shift operators generalizing what done by Kuru Negro as instance in \cite{CKNd}.  
These properties seem to be strictly connected with
separability of the dynamics of the additional degree of freedom, parametrized by the
variable $u$, but they appear to be independent
of the separability within $L$.

\section{Ladder functions and factorized first integrals of extended  Hamiltonians  }\label{s3}

Let $H$ be a Hamiltonian of the form
\begin{equation}\label{Hfor}
H=\frac 12 p_u^2 + \alpha(u) L(p_i,q^i) +\beta (u),
\end{equation}
and let $X_H$, $X_L$ be the Hamiltonian vector fields of $H$ and $L$,
respectively.

\begin{defi}  	 
	The	function $F(p_i,q^j)$ is a \emph{ladder function} of $L$ if there exists a function $f(L)$ such that
	$$
	X_L F=f(L) F.
	$$
\end{defi}

\begin{prop} \label{XL2} A ladder function 
	$F$ of the Hamiltonian $L$  verifies
	(\ref{e1}), i.e.,
	\begin{equation}\label{f00}
	X_L^2F=-2(cL+c_0)F,
	\end{equation}
	if and only if 
	\begin{equation}\label{f0}
	f(L)=\pm \sqrt{-2(cL+c_0)}.
	\end{equation}
\end{prop}

Hence, 

\begin{teo} If $L$ admits a ladder function $G$
with factor $f(L)$ given by (\ref{f0}),
then $L$ admits extensions of the form (\ref{Hest}). Moreover, for  
$$
\alpha=-k^2\gamma', \quad \beta=k^2c_0\gamma^2
+\frac{\Omega}{\gamma^2},
$$
the Hamiltonian (\ref{Hfor}) is an extension of $L$, and the ladder function $G$ generates the first integrals 
(\ref{mn_int}), (\ref{ee2}).
\end{teo} 

The existence of a ladder function of $L$ allows the factorization of the characteristic first integrals $K$ or $\bar K$ of any extension of $L$.

\begin{lem}\label{L1}
	If $X_L G=f(L) G$, then, for any $n\in \mathbb N-\{0\}$, 
	$$X_L G_n= nfG_n,$$ 
	$$
	G_n=(2f)^{n-1}G^n,
	$$
	hence, both $G^n$ and $G_n$ are additional ladder functions of $L$.
\end{lem}

\begin{proof} Recall that
	$$
	G_1=G, \quad G_{n+1}=X_LG G_n+\frac 1n G X_LG_n,
	$$
	then, $G_2=2GX_LG=	2 f(L)G^2$, and
	$$
	X_L G_2=2X_LGX_LG+2GX_L^2G=4G^2f^2= 2f(2G^2f)= 2fG_2.
	$$
	By induction on $n$, we assume that $X_L G_n=nfG_n$ and $G_n=(2f)^{n-1}G^n$, then
	$$G_{n+1}=(fG)G_n+\frac{1}{n}(nfG_n)G=(2fG)G_n=(2f)^n G^{n+1},$$
	$$
	X_LG_{n+1}=X_L( 2fGG_n)=2 (n+1)f^2 GG_n=(n+1)fG_{n+1}.
	$$	
\end{proof}

\begin{teo}\label{GN} 
 Let $G$ be a ladder function of $L$ with factor $f(L)$ given by (\ref{f0}), and
 $H$ be an extension (\ref{Hest}) of $L$ admitting first integrals (\ref{mn_int}) (for $\Omega=0$) or 
 (\ref{ee2}) (for $\Omega\neq 0$). Then,

	\begin{enumerate}
		\item for any  $(m,n)$, the functions 
		(\ref{mn_int}) factorize as
				$$K_{m,n}= (2f)^{n-1}G^nF,$$
		where 
		$$F=\left(p_u+\frac mn\gamma f \right)^m.$$
		\item for any $(s,r)$, the functions 
				(\ref{ee2}) factorize as
		$$\bar K_{2s,r}= (2f)^{r-1}G^rF,$$
		where 
		\begin{align}
		F&=\left[ \left(p_u+\frac{2s}r\gamma f\right)^2+2\Omega \gamma ^{-2}\right]^s \nonumber\\
		&=\left(p_u+\frac{2s}r \gamma f+i\frac {\sqrt{2\Omega}}\gamma  \right)^s \left(p_u+\frac{2s}r \gamma f-i\frac {\sqrt{2\Omega}}\gamma  \right)^s.
		\end{align}
	\end{enumerate}
In both cases, $F$ satisfies the relation
$X_H F=nk^2\gamma'f F.$
\end{teo}

\begin{proof}  By  (\ref{mn_int}) and Lemma \ref{L1} we get
	$$
	K_{m,n}=\left(p_u+\frac m{n^2}\gamma n f \right)^mG_n,
	$$
	that, after simplification, gives the statement.
	Since $K_{m,n}$ is a first integral of $H$, we have
	$$
	0=X_H K_{m,n}=X_H(FG_n)=X_H F G_n+ F X_HG_n.
	$$
Then, because for any function $\phi(q^i,p_i)$ we have $X_H\phi=
	-k^2\gamma'X_L\phi$, we find that $F$ satisfies
	$$
	X_H F=-\frac F{G_n}X_HG_n=k^2\gamma' \frac F{G_n}X_LG_n =k^2 n\gamma'f F.
	$$
	The same reasoning applies to (\ref{ee2}).
\end{proof} 

In \cite{CKNd} the functions $F$ in Theorem \ref{GN} are called \emph{shift functions}.
Clearly, the factor $(2f)^{n-1}$, or $ (2f)^{r-1}$, in the factorization of the first integrals (\ref{ee2}) and (\ref{mn_int}) given in Theorem \ref{GN} is irrelevant and can be dropped down.
 Thus, we have that the first integrals of the extended Hamiltonian (\ref{Hest}) are product of shift and ladder functions.
This fact makes easier the computation of the  algebra of symmetries of $H$ \cite{CKNd}.

\begin{rmk} \rm Due to the expression of the function $f$, the characteristic first integrals  of Theorem \ref{GN}
are not polynomial in the momenta.
\end{rmk}

\section{Extended Hamiltonians and ladder functions }\label{s4}
 Let us assume that $L$ is a $N$-dimensional natural Hamiltonian 
 \begin{equation}\label{L}
 L=\frac 12 g^{ij}p_ip_j +V(q^i).
 \end{equation}
 In the previous section we have shown that 
 a function $G$ satisfying condition (\ref{e1}) for the Hamiltonian $L$ which is also a ladder function for $L$ (or equivalently 
  a ladder function of $L$ with special factor (\ref{f0}))
  allows
  the factorization of the characteristic first integrals $K_{m,n}$ of any extension $H$ of $L$.
  All the
extensions studied until now are based on the existence of solutions $G$ of (\ref{e1}) that are polynomial in the momenta.
 Consequently, the characteristic first integrals of extensions are also polynomial in the momenta. These $G$, however, are not ladder functions of $L$. In this section we show how to obtain ladder functions of $L$ from some of these
 solutions $G$ and, therefore, an alternative factorization of the first integrals is presented.  
Inspired by 	\cite{CKNd}, we search for ladder functions of the form
\begin{equation}\label{gpm1}
G^\pm=A_{\pm}^ip_i+ B_\pm f+ \frac 1f C_\pm,
\end{equation}
with $f=\sqrt{-2(cL+c_0)}$ and ${\mathbf A}_\pm$, $B_\pm$, $C_\pm$, $V$ depending on the $(q^i)$ only. If we are able to determine  ${\mathbf A}_\pm$, $B_\pm$, $C_\pm$, $V$ such that \begin{equation}\label{eqq}
X_LG^\pm=\pm fG^\pm,
\end{equation} 
then we obtain ladder functions $G^\pm$  of $L$ verifying the hypothesis of Theorem \ref{GN}.

\begin{prop}
If we assume $C_\pm=c_1$ constant, then a solution in form (\ref{gpm1}) of (\ref{eqq}) 
is given by
\begin{equation} 	\label{18}
G^\pm=\pm(\nabla G)^ip_i+ G f+ \frac 1f c_1,
\qquad f=\sqrt{-2(cL+c_0)},
\end{equation}
with the function $G$ and the constant  $c_1$ satisfying
\begin{align}
[{\mathbf g},\nabla G]+ c {\mathbf g}G=0,\label{HE1}\\
\nabla V \cdot \nabla G - 2(cV+c_0)G + c_1=0,\label{VE1}
\end{align}
where $[\cdot,\cdot]$ are the Schouten brackets.
\end{prop}

\begin{proof}
By expanding equation (\ref{eqq}) with $G^\pm$ and $L$ given by (\ref{gpm1}) and (\ref{L}) respectively, we have, if $C_\pm$ is constant,
\begin{align*}
& X_LG^\pm\mp fG^\pm = \\
& =p_ip_j(\nabla^iA_\pm^j\pm cg^{ij}B_\pm) +p_i f(\nabla^iB_\pm \mp A^i_\pm)-A^i_\pm \nabla_i V\pm 2(cV-c_0)B_\pm \mp C_\pm.
\end{align*}
By equating to zero the coefficients of the monomials in the momenta, we obtain easily
\begin{align*}
[{\mathbf g},{\mathbf A}_\pm]=\mp c {\mathbf g}B_\pm,\\
\nabla B_\pm=\pm {\mathbf A}_\pm,\\
\nabla V \cdot {\mathbf A}_\pm \mp 2(cV+c_0)B_\pm \pm C_\pm=0,
\end{align*}
 or, equivalently,
\begin{align}
[{\mathbf g},\nabla B_\pm]+ c {\mathbf g}B_\pm=0,\label{HE}\\
\nabla V \cdot \nabla B_\pm - 2(cV+c_0)B_\pm + C_\pm=0,\label{VE}\\
\nabla B_\pm=\pm {\mathbf A}_\pm.
\end{align}
Equations (\ref{HE}), (\ref{VE}) are the same for $(B_+,C_+)$ and $(B_-,C_-)$, thus, their solutions coincide with the solutions $(G,c_1)$ of (\ref{HE1})
and (\ref{VE1}).
Therefore, by setting ${\mathbf A}_\pm=\pm\nabla G$, $B_\pm=G$, and $C_\pm=c_1$, we get (\ref{18}).
\end{proof}

\begin{rmk}\rm \label{r:L0}
If for a given $V$ there exists $G$ satisfying 
 (\ref{VE1}) for some constants $c$, $c_0$, $c_1$,
then the same $G$ is a solution of 
$$
\nabla \tilde{V}
 \cdot \nabla G - 2(\tilde{c} \tilde{V}+\tilde{c}_0)G + \tilde{c}_1=0,
$$  
for  the potential $\tilde{V}= a_1V+a_2$ ($a_1,a_2\in\mathbb{R}$) by choosing
$\tilde{c}=c$, $\tilde{c}_1=a_1c_1$, $\tilde{c}_0=a_1c_0+c a_2$.
This means that if a solution $G$ is known for $c_0=0,$ $c\neq 0$, then, by adding the inessential constant
$\tilde{c}_0/c$
to the potential $V$ we get a solution $G$ of (\ref{VE1}) for
any value of $\tilde{c}_0$.
\end{rmk}

Equations (\ref{HE1}) and (\ref{VE1}) must be compared with equation (\ref{e1})
$$
X_L^2(G)=-2(cL+c_0)G,
$$
which determines the function $G$ necessary to build the extensions of $L$. By assuming $G(q^i)$, equation (\ref{e1}) splits into \cite{CDRfi}
\begin{align}
\nabla_i\nabla_jG+cg_{ij}G=0,\label{he}\\
\nabla V \cdot \nabla G-2(cV+c_0)G=0.\label{ve}
\end{align}
It is clear that
 (\ref{HE1}) and (\ref{VE1}) are equivalent 
 to 
 \begin{equation}\label{condB}
 X_L^2(G)=-2(cL+c_0)G + c_1,
 \end{equation} 
since $[g,\nabla G]=[g,[g,G]]$. Thus, $G$ must verify the hessian equation (\ref{he}).  
Hence, for $c_1=0$, we get (\ref{he}) and (\ref{ve}). This means  that,  given any function $G$ solution of (\ref{he}),
(\ref{ve}), we can immediately obtain solutions $G^\pm$  (\ref{18}) just by putting  $c_1=0$:
\begin{prop} \label{p8}
Let $G(q^i)$ be any solution of (\ref{he}), (\ref{ve}) (i.e. of $X_L^2G=f^2G$), then, the functions 
$$
G^\pm=\pm\nabla^i Gp_i+ G f,
$$
are  ladder functions of $L$, i.e. solutions of  
	$$
	X_LG^\pm=\{L,G^\pm\}=\pm fG^\pm.
	$$
Moreover, the extended Hamiltonian $H$ of $L$ admits the factorized first integrals described in Theorem \ref{GN}.	
\end{prop}

The  solutions $G$ of (\ref{he}) have been computed for several Riemannian and pseudo-Riemannian manifolds in  \cite{CDRPol}, \cite{CDRfi}, \cite {CDRgen}, \cite{CDRsuext}. Moreover, in \cite{CDRfi} it is proved that the constant $c$ is related to the curvature, or the sectional curvatures, of the base manifold of $L$.
Furthermore, equation (\ref{ve}) can be seen as a compatibility condition that the potential $V$ has to fulfil.
By admitting $c_1\neq 0$, (i.e., considering  
 (\ref{VE}) instead of (\ref{ve}) as compatibility condition of the potential) we enlarge the class of compatible  potentials $V$, which would depend on $c_1$ 
 and which reduces for $c_1=0$ to the class of potential verifying (\ref{ve}). 

However, this technique does not  imply  that we are enlarging  the set of Hamiltonians that admit extensions.
Indeed, more general functions $G(q^i,p_i)$ which are polynomial in the momenta and solutions of $X_L^2G=f^2G$ on the same manifold, can be compatible with these potentials $V$, as is the case of the TTW system (see \cite{TTWcdr} and \cite{CKNd}), where $G$ is linear and homogeneous in the momenta.
In this case, for example, we have $V=(C_1+C_2 \cos q)\sin ^{-2}q$, and we know the solution $G=p \sin q$ for $c=1$ and $c_0=0$. However, 
for the same values of $c$ and $c_0$, the (\ref{HE1}), (\ref{VE1}) can be integrated, yielding the solution $G=\cos q$ and $c_1=-C_2$.

\begin{rmk}\rm 
Since any ladder function of $L$ is also a solution of (\ref{e1}), the solutions $G^\pm$  of the form (\ref{18}),  for a given $L$, allow the extension of the Hamiltonian $L$ by putting $G=G^\pm$. We observe that the conditions on $c$, i.e. constraints on the curvature or sectional curvature of the base manifold of $L$, remain unchanged.
\end{rmk}

\begin{exm} \rm
In \cite{CDRfi} are computed the extensions for natural Hamiltonian $L$ on  $\mathbb S^3$ with Hopf coordinates $(q^1=\eta, q^2=\xi_1, q^3=\xi_2)$. In these coordinates the non-zero metric components are $g_{11}=1$, $g_{22}=\sin ^2 \eta$ and $g_{33}=\cos^2 \eta$. The function $G$, obtained for $c=1$ equal to the curvature of the sphere and $c_0=0$, is
$$
G=(a_3\sin \xi_1+a_4 \cos \xi_1)\sin \eta+(a_1\sin \xi_2+a_2\cos \xi_2)\cos \eta,
$$
where $a_i$ are arbitrary constants.
The corresponding ladder functions are
\begin{align*}
	G^\pm=\pm \left\{ \left[(a_3\sin \xi_1+a_4 \cos \xi_1)\cos \eta-(a_1\sin \xi_2+a_2\cos \xi_2)\sin \eta) \right]p_{\eta}+\right.\\
	\left. \left[\frac {a_3\cos \xi_1-a_4\sin \xi_1}{\sin \eta}\right]p_{\xi_1}+\left[\frac {a_1\cos \xi_2-a_2\sin \xi_2}{\cos \eta}\right]p_{\xi_2}\right\}+G f+\frac{c_1}{f}.
\end{align*}
In \cite{CDRfi}, it is provided 
 a solution $V$ of (\ref{ve}) for 
 $a_1=1$, $a_2=a_3=a_4=0$, 
  as
$$
V=\frac 1{\sin ^2 \eta}F\left(\xi_1, \frac {\tan \eta}{\cos \xi_2}\right),
$$
where $F$ is any smooth function, potential that can be separable or not in the chosen coordinates. With the same assumptions, the integration of (\ref{VE1}) for $c_1\neq 0$ gives
$$
V=\frac 1{\sin ^2 \eta}F\left(\xi_1, \frac {\tan \eta}{\cos \xi_2}\right)+ \frac {c_1\cos \eta \sin \xi_2}{1-\cos ^2\eta \sin ^2\xi_2},
$$
in accordance to the fact that, for $c_0=0$,
equation (\ref{ve}) is a linear homogeneous differential
equation in $V$ and a solution of (\ref{VE1})
is the sum of the solution of the homogeneous 
equation with a particular solution of the non-homogeneous case. 
\end{exm}

\section{Comparison with Kuru-Negro factorization}
\label{s5}

In \cite{CKNd} the authors determine for the classical Hamiltonian  of the TTW system in polar coordinates $(r,\theta)$ shift and ladder functions, from which  constants of the motion $X^\pm$ are obtained depending on a rational parameter $k=m/n$:
$$
X^\pm=(\Xi^\pm)^n(A^\mp)^{2m}(D^\pm)^m.
$$
Their method is: (i) to split the Hamiltonian into a radial part combined together with a one-dimensional Hamiltonian $H_\theta$
(the Posch-Teller Hamiltonian);
(ii) to determine  ladder functions
$\Xi^\pm$ of $H_\theta$ with factor $\mp\sqrt{-8H_\theta}$;
(iii) to construct the functions $A$ and $D$ 
interconnecting  the radial part 
with the constant of the motion $k\sqrt{H_\theta}$.

We show that this construction of a factorized first integral of $H$ can be straightforwardly generalized to any $H$ which is an extension
(\ref{Hest}) with $c\neq 0$ of a Hamiltonian $L$ admitting
 ladder functions $G^\pm$ such that
 \begin{equation}\label{gpm}
 X_LG^\pm=\pm\sqrt{-2(cL+c_0)}G^\pm.
 \end{equation}
However, while Theorem \ref{GN} holds for any extension (also for $c=0$), in the following construction the hypothesis that $c\neq 0$ seems 
to be crucial.

By using the property (\ref{eqgam}) of 																																																																																																																																																																																																																																																																																																																																																																																																																																																																																																																																																																																																																																																																																																																																																																																								the $\gamma$, we can write any extension (\ref{Hest}) of $L$ with $c\neq 0$ as
\begin{equation}\label{H}
H =\frac 12 p_u^2+\frac 12 \left(\gamma^2+\frac{C}{c}\right)M^2+\frac{\omega^2}{2\gamma^2}-\frac{C}{c} k^2c_0,
\end{equation}
where, in order to underline the similarity with the formulas in \cite{CKNd}, we set $M=k\sqrt{2(cL+c_0)}$
and $\omega^2=2\Omega$.

Therefore, we have 

\begin{prop}\label{teo}
	Given a Hamiltonian $L$ such that equation (\ref{gpm}) admits a solution $G^\pm$ for some constants $c$ and $c_0$ not both zero, then, 	
	(i) $L$ admits extensions, 
	(ii) for any extension $H$ (\ref{H}) of 
	$L$,  the functions
	$$
	X^\pm=(G^\pm)^{2n}(A^\mp)^{2m}(D^\pm)^m,
	$$
	with 
	$$
	A^\pm=\mp \frac i{\sqrt{2}} p_u+\frac{\omega}{\sqrt{2} \gamma} -\frac M{\sqrt{2}}\gamma,
	$$
	$$
D^\pm=\mp \frac i{\gamma} p_u+\frac{\omega}{\gamma^2} -\frac{H}{\omega}-
\frac{C}{c\omega}\left(c_0k^2- \frac{M^2}{2} \right),
	$$
are factorized first integrals of $H$.
\end{prop}

\begin{proof} 
	The Hamiltonian $L$ admits an extension $H$ because of Proposition \ref{XL2}. Moreover, we have, for any solution $G^\pm$ of (\ref{gpm})
	\begin{align*}
	\{H,G^\pm\}&=
	k^2(c\gamma^2+C)\{L,G^\pm\}
	\\
	&=
	\pm k^2(c\gamma^2+C)\sqrt{-2(cL+c_0)} G^\pm
	\\ &=
	\mp ik \gamma' MG^\pm,
	\end{align*}
	and, since $k=m/n,$
		\begin{equation}\label{HG}
	\{H,(G^\pm)^{2n}\}=\mp i 2m	\gamma' M (G^\pm)^{2n}.
		\end{equation}
The functions $A^\pm$ and $D^\pm$ satisfy the
following factorization and commutation formulas
	$$
	H=A^+A^-+\omega M +\frac{C}{c}\left(\frac {M^2}2- k^2d\right),
	$$
	$$
	D^+D^-
	-\left(\frac{H}{\omega}+
	\frac{C}{c\omega}\left(c_0k^2- \frac{M^2}{2} \right)\right)^2=- M^2,
	$$
\begin{equation}\label{HA}
	\{H,A^\pm \}=\mp i\gamma'\left(M+\frac \omega{\gamma^2}\right)A^\pm,
\end{equation}
\begin{equation}\label{HD}
	\{H,D^\pm\}=\mp 2i\gamma' \frac{\omega}{\gamma^2}D^\pm. \end{equation}
The statement follows by computing directly $\{H,X^\pm\}$ taking into account  (\ref{HA}),(\ref{HD}) and (\ref{HG}).
\end{proof}

We recall that, from  Theorem \ref{GN} and Proposition \ref{p8}, it follows  that the factorization of the characteristic first integral is, for $k=\frac mn$, of the form $(G^\pm)^nF^m$ (up to the factor  $f^{n-1}$, inessential because  itself first integral of $H$). Therefore, the factorization determined by Theorem \ref{GN} appears to be simpler than the one given in Proposition \ref{teo}.

\begin{exm}\rm
 The ladder functions $\Xi^\pm$ of \cite{CKNd} for the Tremblay-Turbiner-Winternitz  Hamiltonian satisfy (\ref{f00}) and guarantee the extensibility of the Hamiltonian $H_\theta$. Indeed, given (26) of \cite{CKNd} 
$$
\{H_\theta,\Xi^\pm \}=\pm 4 i\sqrt{H_\theta}\Xi^\pm,
$$
with
$$
\Xi^\pm=\pm i p_\theta \sin 2\theta +\sqrt{H_\theta} \cos 2\theta +\frac{\beta^2-\alpha^2}{\sqrt{H_\theta}},
$$
$$
H_\theta=p_\theta^2+\frac{\alpha^2}{\cos ^2 \theta }+\frac{\beta^2}{\sin ^2 \theta },
$$
and
$$
\{\theta, p_\theta\}=-1,
$$
as we usually assume, we have that the $G^\pm=\Xi^{\pm}$  are solutions  of
$$\{L,G\}=\pm \sqrt{-2(cL+c_0)}G,
$$
and therefore of
$$\{L,\{L,G\}\}=-2(cL+c_0)G,$$ 
with $L=H_\theta$, $c=8$ and $c_0=0$ (sec 3.2 of \cite{CKNd}).
Up to a factor 1/2, the radial term of $H$ in \cite{CKNd} coincides with the $(u,p_u)$-dependent term in the extensions 
(\ref{H}), in the case $\gamma =1/r$
($C=0,c=1$) and $\omega^2=2\Omega$.
\end{exm}

\section{Shift and Ladder Operators for Quantum Extended Hamiltonians}\label{s6}
Our aim is to give a general method for constructing a $(N+1)$ quantum Hamiltonian operator admitting a factorized symmetry, starting from 
a natural Hamiltonian
$L$ with $N$ degrees of freedom which admits a ladder function of the form (\ref{gpm1}), i.e., such that there exists a non trivial solution of  
(\ref{HE1}) and (\ref{VE1}).
As a first step, we look for shift and ladder operators similar to those given in \cite{CKNd}.
Finally, as in \cite{CKNd}, we use them in order to construct a symmetry operator for the extended quantum Hamiltonian.
The structure of this Hamiltonian operator
is strictly related with the classical extension  of $L$,
but does not exactly coincide with the Laplace-Beltrami quantization of an extension of
$L$ \cite{CDRmlbq}.

\begin{defi} Let $\hat H_M$ be an operator depending on a parameter $M$. We say that the $M$ dependent operator  $\hat{\mathcal{S}}_M$ is a \emph{shift operator} for $\hat{H}_M$ if there exists a
function $\bar M=\bar{M}(M)$  such that
	$$
	\hat{\mathcal{S}}_M \hat H_M=\hat H_{\bar{M}}\hat{\mathcal{S}}_M.
	$$ 
Let $\hat H$ be an operator and the function $\psi_\lambda$ be an eigenfunction of $\hat{H}$  with eigenvalue $\lambda$. We say that $\hat{\mathcal{L}}_\lambda$ is a \emph{ladder operator} for the eigenfunctions of $\hat H$ if there exists a function
 $\bar{\lambda}=\bar{\lambda}(\lambda)$ such that
$$
\hat H \psi_\lambda =\lambda \psi_\lambda \ \Rightarrow \ 
\hat{H}(\hat{\mathcal{L}}_\lambda \psi_\lambda) =  \bar{\lambda} \hat{\mathcal{L}}_\lambda\psi_{\lambda}.
$$ 
If $\bar{\lambda}(\lambda)=\lambda$ for all eigenvalues $\lambda$ we say that $\hat{\mathcal{L}}_\lambda$ is a \emph{symmetry} for
$\hat{H}.$	\end{defi}

We consider the Schr\"odinger operator associated with the natural Hamiltonian $L$ (\ref{L}), 
\begin{equation}\label{hL}
\hat L=-\frac {\hbar^2}2\Delta +V. 
\end{equation}
where $\Delta$ is the Laplace-Beltrami operator associated to the metric $\mathbf g$.
Let $\mathbf{\bar g}$ be  the
metric tensor of the extension of (\ref{L})  which is of the form
\begin{equation}
\bar g^{ij}=\left(\begin{matrix} 1 & 0 \cr 0 & \alpha(u)g^{ij} \end{matrix} \right),\quad \alpha(u)=\frac {m^2}{n^2}\left(c\gamma^2+{C}\right)=-k^2\gamma',
\end{equation}
since $\gamma$ satisfies (\ref{eqgam}). 
We choose $\mathbf{\bar g}$ as metric tensor  for the quantum extended Hamiltonian of (\ref{hL}).
Therefore, the Laplace-Beltrami operator $\bar \Delta$ associated with the metric $\mathbf{\bar g}$ is
\begin{equation}
\bar \Delta \psi=\partial ^2_u \psi+ N c\gamma \partial _u \psi-
k^2\gamma' \Delta \psi,
\end{equation}
where $N$ is the dimension of the base manifold of $L$.
We consider an extended Schr\"odinger operator of
$(\ref{hL})$ of the form
\begin{equation}\label{HHH}
\hat H=-\frac{\hbar^2}2 (\partial ^2_u \psi+ N c\gamma \partial _u \psi-
k^2\gamma' \Delta \psi) -k^2\gamma' V+ W(u). \end{equation}
where $k$ is a constant and the function $W$ will be suitably determined later.
If we consider eigenfunctions $\psi_\lambda$ of the operator $\hat{L}$
$$
\hat L\psi_\lambda=\lambda \psi_\lambda,
$$
then, for any $\phi=\phi(u)$,
$$
\hat H(\phi\psi_\lambda) =-\frac{\hbar^2}2\left(\partial_u^2 + Nc \gamma \partial_u\right)(\phi)\psi_\lambda+
(W-k^2\gamma'\lambda)\phi\psi_\lambda.
$$ 
Since by (\ref{eqgam}) $\gamma'=-c\gamma^2-C$, we may write
\begin{equation}\label{H1}
\hat H(\phi\psi_\lambda) =-\frac{\hbar^2}2\left(\partial_u^2 + Nc \gamma \partial_u\right)(\phi)\psi_\lambda +
(W+ck^2\lambda\gamma^2+ C k^2\lambda)
\phi\psi_\lambda.
\end{equation}
Hence, we have that $\hat{H}(\phi\psi_\lambda)=
E\phi\psi_\lambda$
(that is $\phi\psi_\lambda$ is an eigenfunction of (\ref{HHH}) associated with the eigenvalue $E$) if and only if $\phi(u)$ is 
an eigenfunction of the 
$M$-dependent operator 
\begin{equation}\label{hHM}
\hat{H}^M=
 -\frac{\hbar^2}2\left(\partial_u^2 + Nc \gamma(u) \partial_u\right) 
+M(c\gamma^2(u)+C)+W(u), 
\end{equation}
with $M=k^2\lambda$. 
We denote such a $\phi$ by $\phi^{k^2\lambda}_E$ in order to underline its dependence on the spectral parameters $E$ and $\lambda$. 

\begin{defi} \rm \label{d:ws}
We say that $\hat{X}$ is a \emph{warped symmetry operator} of 
(\ref{HHH}) if, for all partially separated eigenfunctions $f_E=\phi^{k^2{\lambda}}_E\psi_\lambda$, $\hat{H}\hat{X}(f_E)=E\hat{X}(f_E)$, or equivalently, that $[\hat{H},\hat{X}](f_E)=0$.
We allow $\hat{X}$ to depend on both $\lambda$ and 
$E$.
\end{defi}

In the following, we construct operators $\hat G$, $\hat A$, $\hat D$ which are the factors of a warped symmetry $\hat X$ of $\hat H$ for $k\in \mathbb{Q}$.
In particular, with Theorem \ref{teoq} below, we produce a quantum version of the factorized characteristic first integrals of extended Hamiltonians of Proposition \ref{teo}, valid only for the case $c\neq 0$, $C=0$. The proof of the Theorem is given in Sec.\ \ref{6.4}, while in Sec.\ \ref{6.1}-\ref{6.3} are studied the relevant properties of the operators $\hat G$, $\hat A$ and $\hat D$, respectively.

Given any function $G(q^1,\ldots,q^N)$, we define
the operators \begin{equation}\label{dhG}
\hat G^\pm_\epsilon=\aa \nabla^iG \nabla_i+a_1^\pm G+a_2^\pm
\end{equation}
where the constants $a_i$ verify the relations
\begin{equation}\label{be}
 a_1^\pm=\aa\left(\frac {c(1-N)}2\pm \epsilon \frac {\sqrt{2|c|}}\hbar\right), \qquad
a_2^\pm=\frac{-2c_1\aa}{-c\hbar^2\pm 2\hbar\epsilon\sqrt{2|c|}},\end{equation}
with $\epsilon, c_1\in \mathbb{R}$, $c\neq 0$.
Moreover, for any integer $n>0$, we construct the operators
$(\hat{G}^{\pm}_\epsilon)^n$ as follows
\begin{equation}\label{hGn}
(\hat{G}^{\pm}_\epsilon)^n= \hat{G}^{\pm }_{\epsilon\pm(n-1)s}\circ \hat{G}^{\pm }_{\epsilon \pm(n-2)s}\circ \cdots \circ \hat{G}^{\pm }_{\epsilon \pm s}\circ \hat{G}^{\pm }_{\epsilon} 
\end{equation}
where $s=\sqrt{\frac{|c|}{2}}\hbar,$ and $\hat{G}^\pm_\epsilon$ is defined in (\ref{dhG}).

We define the operators
\begin{equation}\label{dhA}
\hat{A}^{\sigma,\tau}_\mu= \partial_u+\frac{1}{\hbar} \left[\sigma  {\omega} cu+\frac{\hbar c(N-1)+\tau\sqrt{2|c|}\mu}{2cu}
\right], \qquad (\sigma,\tau=\pm 1),
\end{equation}
\begin{equation}\label{Dhat}
\hat D_E^\pm= \frac {\hbar}{\sqrt{2}}u\partial_u
+\frac{\hbar (N+1)}{2\sqrt{2}}\pm\left( \frac{E}{2c\omega}-\frac{\omega}{\sqrt{2}} cu^2\right),
\end{equation}
where $E$, $\mu$, $\omega\in \mathbb{R}$, $c\neq 0$.
Moreover, for any integer $m>0$, we define the operators 
\begin{equation}\label{dAm}
(\hat{A}^{\sigma,\tau}_\mu)^m=\hat{A}^{\sigma,\tau}_{\mu +\tau(m-1)s}\circ \hat{A}^{\sigma,\tau}_{\mu + \tau(m-2)s}\circ \cdots \circ \hat{A}^{\sigma,\tau}_{\mu + \tau s}\circ \hat{A}^{\sigma,\tau}_{\mu}, 
\end{equation}
\begin{equation}\label{dDm}
(\hat{D}^{\pm}_E)^m=\hat{D}^{\pm }_{E \pm 2(m-1)\delta} \circ \hat{D}^{\pm }_{E \pm 2(m-2)\delta}\circ \cdots \circ \hat{D}^{\pm}_{E \pm 2\delta}\circ \hat{D}^{\pm}_{E}, 
\end{equation}
with $s=\sqrt{\frac{|c|}{2}}\hbar$, $\delta=\hbar c\omega,$  $\sigma,\tau=\pm 1$.
 
 In (\ref{hGn}) and (\ref{dDm}), the signs in upper
and lower indices must be the same (all + or all $-$).

We are able to state the theorem as follows 

\begin{teo}\label{teoq}
Let $L_0=\frac{1}{2}g^{ij}p_ip_j+V_0(q^i)$ be a classical natural Hamiltonian with $N$ degrees of freedom 
and $G$ a function of $(q^i)$ satisfying (\ref{HE1}), (\ref{VE1}) for  $c\neq 0$ and $c_0=0$
i.e., generating 
 a ladder function of the form (\ref{18}) for $L_0$.
For any $k\in \mathbb{Q}^+$, let us construct the operator
 \begin{equation}\label{hX}
 \hat{X}^k_{\epsilon,E}=(\hat{G}^{+}_\epsilon)^{2n} \circ
 (\hat{A}^{1,1}_{k\epsilon})^{2m}\circ (\hat{D}^{+}_E)^m
 \end{equation}
 where $m$,$n$ are positive integers  such that $k=m/n$, the operators
 $(\hat{G}^{+ }_\epsilon)^{2n}$, $(\hat{A}^{\sigma,\tau}_\mu)^{2m}$ (with $\mu=k\epsilon$, $\sigma=\tau=1$) and $(\hat{D}^{+}_E)^m$ are defined in (\ref{hGn})
 (\ref{dAm}) (\ref{dDm}) and 
 where $\hat{G}_\epsilon^\pm$ are constructed from the function $G$.
 Then, the
 quantum Hamiltonian 
\begin{equation}\label{hL0}
\hat{L}_0= -\frac{\hbar^2}2\Delta+V_0,
\end{equation}
associated with $L_0$, can be extended 
to the quantum Hamiltonian
\begin{equation}\label{Hokk}
\hat{H}=-\frac{\hbar^2}{2}\left(\partial^2_{uu}+ \frac{N}{u}\partial_u
\right)+\frac{k^2}{cu^2}\hat{L}_0 + \frac{\omega^2}{2}c^2u^2+ \frac{\dok (k^2+1)}{cu^2},
\end{equation}
with
\begin{equation}\label{dok}
\dok=-\frac{\hbar^2c(N-1)^2}{8},
\end{equation}
which admits the operator $\hat{X}^k_{\epsilon, E}$ as warped symmetry: i.e., 
\begin{equation*}\label{hHhX}
\hat{H}(\hat{X}^k_{\epsilon, E} f_E)= E\hat{X}^k_{\epsilon, E} (f_E),
\end{equation*}
for all (partially) multiplicatively separated eigenfunctions $f_E=\phi_E^{M}(u)\psi_\lambda(q^i)$ of $\hat{H}$ such that $\psi_\lambda$ is an eigenfunction of $\hat{L}_0$ with eigenvalue $\lambda$,  and
$$M=k^2(\lambda+\dok),\qquad \epsilon
=\sqrt{| \lambda +\dok|}, \qquad
c(\lambda+\dok)\ge 0.$$
\end{teo}

\begin{rmk}
\rm
For $N=1$, the quantum extended Hamiltonian (\ref{Hokk}) is
 the quantum operator associated with the
extension of $L_0$ with $\gamma=1/(cu)$ (that is for $c\neq0$, $C=0$).
However, for $N>1$ (that is when $L_0$ has more than one degree of freedom) this is no more true, due to 
 the presence of the last 
addendum in (\ref{Hokk}). 
The addendum could be absorbed in $L_0$, by replacing $L_0$ with
${L}_1=L_0+\dok\frac{k^2+1}{k^2}$.
Nevertheless, in the proof of Theorem \ref{teoq} it is
essential to consider the eigenvalues of $\hat{L}_0$, not those of the shifted operator $\hat{{L}}_1$. 
From a different view point, the additional term proportional to $\dok$  added to
the quantization of the classical extension of $L_0$
suggests a link with the introduction of a kind of quantum correction, used as instance in \cite{CDRmlbq, B2} to allow quantization for classical quadratic in the momenta first integrals.
\end{rmk}

\begin{rmk}
\rm
Other choices of the signs in $\hat{G}^\pm_\epsilon$, $\hat{A}_\mu^{\sigma,\tau}$,  and $\hat{D}^\pm_E$
(as well as the commutation between $(\hat{G}^\pm_\epsilon)^{2n}$ and $(\hat{A}^{\sigma,\tau}_{k\epsilon})^{2m}\circ (\hat{D}^{+}_E)^m$ ) allow us to construct other warped symmetry operators for (\ref{Hokk}):
by applying the same kind of  shown given in Section \ref{6.4}, it easy to check that the  operators  
$$(\hat{G}^{-}_\epsilon)^{2n}
(\hat{A}^{-1,-1}_{k\epsilon})^{2m}(\hat{D}^{-}_E)^m, \quad (\hat{G}^{-}_\epsilon)^{2n}
(\hat{A}^{1,-1}_{k\epsilon})^{2m}(\hat{D}^{+}_E)^m, \quad (\hat{G}^{+}_\epsilon)^{2n}
(\hat{A}^{-1,1
}_{k\epsilon})^{2m}(\hat{D}^{-}_E)^m.$$
commute with (\ref{Hok}), on separated eigenfunctions.
\end{rmk}

\begin{rmk} \rm
Shift and ladder operators arise usually from the factorization, complete or partial, of the Hamiltonian operator. While this property is useful for computations, we did not make use of it in the construction of the symmetry operators or of the first integrals.
\end{rmk}

\subsection{Ladder operators for the eigenfunctions of $\hat L$}\label{6.1}

\begin{prop}\label{p:Ghat}
Let $G$ be solution of (\ref{HE1}) and (\ref{VE1}), with $c \neq  0$ and 
\begin{equation}\label{cons}
c_0=-\frac{c^2\hbar^2}8(N-1)^2.
\end{equation} 
 By setting, for all $\lambda$ such that $c\lambda\ge 0$,
$$
\epsilon=\sqrt{|\lambda|} \qquad (i.e,\ \lambda=\check{c}\epsilon^2 \ with\ \check{c}=sgn(c)
),
$$
then,
for any $\aa\neq 0$, the operators $\hat G_\epsilon^\pm$ (\ref{dhG}) are  ladder operators for the eigenfunctions of the Hamiltonian operator (\ref{hL}): i.e., for any $\psi_\lambda$ such that
$\hat{L}\psi_\lambda=\lambda\psi_\lambda$
\begin{equation}\label{40}
\hat L\hat G^-_\epsilon\psi_\lambda=\check{c}\left(\epsilon+ s\right)^2\hat{G}_\epsilon^-\psi_\lambda,
 \quad
\hat L\hat G_\epsilon^+\psi_\lambda=\check{c}\left(\epsilon- s
\right)^2\hat{G}_\epsilon^+\psi_\lambda,
\end{equation}
with $s=\hbar\sqrt{\frac{|c|}{2}}$.
\end{prop}

The proof of the above proposition is based on several lemmas.

\begin{lem}\label{l:1} 
Let $\hat G=\aa \nabla^iG \nabla_i+\bb G+\ee$, where $G$ is solution of (\ref{HE1}), and $\aa,\bb,\ee$ are constants. We have for any function $\psi$
	\begin{align}
	\Delta \hat G (\psi)&= \aa\nabla^iG g^{lm}\nabla_l\nabla_m \nabla_i \psi+(2\bb-\aa c)\nabla^iG\nabla_i \psi 
	\\ \nonumber
	& +G(\bb-2\aa c)\Delta \psi -\bb cNG \psi+\ee \Delta\psi.
	\end{align}
	\end{lem}
	\begin{proof}
Direct computation, taking into account that solutions $G$ of (\ref{HE1}) verify (\ref{he}) and its contraction with the contravariant metric tensor: $\Delta G=-cNG$.
		\end{proof}

\begin{lem} We have for any function $\psi$
\begin{equation}
(g^{jk}\nabla_k\nabla_j\nabla_i -\nabla_i \Delta) \psi =-R^l_{\ i}\nabla_l \psi,
\end{equation}
where $R^l_{\ i}$ are the components of the Ricci tensor.
\end{lem}

\begin{proof} It follows from the direct computation, using the definition of $R^l_{\ ijk}$ as
	$$
R^l_{\ ijk}=\partial_j\Gamma^l_{ik}-\partial_k\Gamma^l_{ij}+\Gamma^m_{ik}\Gamma^l_{mj}-\Gamma^m_{ij}\Gamma^l_{mk},
	$$
the relation
	$$
\partial_ig^{jk}=-g^{jl}\Gamma ^k_{il}-g^{lk}\Gamma ^j_{il},
	$$
and the equations $R_{ri}=R^j_{\ rij}$, $R^l_{\ ijk}=-R^{\ l}_{i\ jk}$.
\end{proof}

	\begin{lem} If $G$ is a solution of (\ref{HE1}) then
	$$
	R_{ij}\nabla^jG=c(1-N)\nabla_iG.
	$$
		
	\end{lem}
	
	\begin{proof} By differentiating $\nabla_i\nabla_jG=-c g_{ij}G$ we have (see also \cite{CDRfi})
		\begin{align*}
 		c(g_{jl}\nabla_iG-g_{il}\nabla_jG)=\nabla_kGR^k_{lij}=\nabla^kGR_{klij}=-\nabla^kGR_{lkij},
		\end{align*}
		by contracting the first and the last terms with $g^{lj}$
		we have the statement.
\end{proof}

	Then, it follows
	
	\begin{lem}\label{nabladelta} Let us consider $\hat G=\aa \nabla^iG \nabla_i+\bb G+\ee$, where $G$ is solution of (\ref{HE1}) and $\aa,\bb,\ee$ are constants. For any function $\psi$ 	we have
		\begin{align}
		\Delta \hat G (\psi)&= \aa\nabla^iG \nabla_i\Delta \psi+(2\bb-\aa c(2-N))\nabla^iG\nabla_i \psi -\bb cNG \psi \\ \nonumber 
		& \qquad +
		G(\bb-2\aa c)\Delta \psi+\ee \Delta \psi.
		\end{align}
				\end{lem}
Under the assumptions of Lemma \ref{nabladelta}  we can  easily evaluate the operator $\Delta \hat G$ on an eigenfunction $\psi_\lambda$ of $-\frac {\hbar^2}2 \Delta+V$ with eigenvalue $\lambda$ i.e., such that
$$
\Delta \psi_\lambda=\frac 2{\hbar^2}(V- \lambda) \psi_\lambda.
$$
\begin{lem} Let be $\hat G=\aa \nabla^iG \nabla_i+\bb G+\ee$, where $G$ is solution of (\ref{HE1})-(\ref{VE1}), with $\aa,\bb,\ee$ constants. We have
	\begin{align} \label{46}
	\left(-\frac {\hbar^2}2\Delta +V\right)\hat G (\psi_\lambda)=\left[\frac{\hbar^2}2(\aa c(2-N)-2\bb)+\aa\lambda\right]\nabla^iG\nabla_i\psi_\lambda+\\ \nonumber +  \left[\lambda(\bb-2\aa c)+\frac{\hbar^2}2c\ee N-2\aa c_0\right]G\psi_\lambda+(\aa c_1+\ee\lambda)\psi_\lambda.
	\end{align}
		
\end{lem}

We now require that $\hat G \psi_\lambda$ is an eigenfunction of $\hat{L}=-\frac{\hbar^2}2\Delta+V$. By (\ref{46}), the condition is clearly equivalent to the equations
\begin{align*}
\left[\frac{\hbar^2}2(\aa c(2-N)-2\bb)+\aa\lambda\right]\bb=\left[\lambda(\bb-2\aa c)+ \frac{\hbar^2}2c\bb N -2\aa c_0\right]\aa,\\
\left[\frac{\hbar^2}2(\aa c(2-N)-2\bb)+\aa \lambda\right]\ee
=\aa^2 c_1+\aa \ee\lambda.
\end{align*}
Therefore,
\begin{lem} \label{l:u} The function $\hat G \psi_\lambda=
(\aa \nabla^iG \nabla_i+\bb G+\ee)\psi_\lambda$ is an eigenfunction of $\hat{L}=-\frac{\hbar^2}2\Delta+V$ if and only if the constants $\aa,\bb,\ee$ satisfy
	\begin{align}\label{beq}
	2(c\lambda+c_0)\aa^2+c{\hbar^2}(1-N)\aa\bb-\hbar^2\bb^2=0,\\ \label{eeq}
	\frac{\hbar^2}2  c(N-2)\aa \ee + \hbar^2 \bb \ee +\aa^2 c_1=0.
	\end{align}
\end{lem}
If $\aa=0,$ then (\ref{beq}) implies $\bb=0$, while equation (\ref{eeq}) is satisfied for all $\ee$. However in this case the operator $\hat{G}$ reduces to  multiplication by a constant (a trivial symmetry operator).
Hence, we consider $\aa\neq 0$.

For $c\neq 0$, from (\ref{beq}), we deduce
$$
\bb= \aa\left[\frac{c}{2}(1-N)\pm \frac{1}{2\hbar} 
\sqrt{c^2\hbar^2(N-1)^2+8(c\lambda+c_0)}\right]. 
$$
which is real if $c\lambda+c_0\ge -c^2\hbar^2(N-1)^2/8$.
If in (\ref{eeq}) the coefficient of $\ee$ vanishes,
then the solution of (\ref{beq}--\ref{eeq}) is $\aa=\bb=0$ already considered.
Otherwise, the value for the parameter $\ee$ satisfying (\ref{eeq}) is
$$
\ee=\frac{2c_1\aa}{c\hbar^2\mp \hbar\sqrt{c^2\hbar^2(N-1)^2+8(c\lambda+c_0)}}
$$
and the corresponding eigenvalue for the eigenfunction
$\hat{G}\psi_\lambda$ is
$$
\bar{\lambda}=\lambda+\frac{\hbar^2c}{2}\mp \frac{\hbar}{2}
\sqrt{c^2\hbar^2(N-1)^2+8(c\lambda+c_0)}.
$$
\begin{proof1}{Proposition \ref{p:Ghat}}
In \cite{CKNd} the ladder operators depend on a parameter $\epsilon$ and transform the eigenvalue $\lambda=\epsilon^2$ into $\bar{\lambda}=(\epsilon\pm 2)^2$.
If we want that the eigenvalue associated with $\hat{G}\psi_\lambda$ is of the form $$\bar{\lambda}=\pm(\sqrt{|\lambda|}+ s)^2,$$ for a constant $s$, then (after Lemmas \ref{l:1}--\ref{l:u})
the value of $c_0$ has to be (\ref{cons}), the condition $c\lambda\ge 0$ has to be satisfied
and the values of $\bb$, $\ee$ and $\bar{\lambda}$
are those appearing in (\ref{be}) and (\ref{40}), respectively.
\end{proof1}
 
 \begin{rmk}\rm 
It is easy to check what happens to 
 $\hat{G}_\epsilon^\pm$ in the particular case $N=1$, $V=c_0=c_1=0$: in accordance with the first equation  (9) and relations (10) in \cite{CKNd}, for $c=4$ and $G=\cos 2\theta$, after setting $\hbar^2=2$, we obtain the solution
$\bb=\pm 2\epsilon \aa$, $\ee=0$ with the eigenvalue $\bar{\lambda}=(\epsilon\pm 2)^2$.
\end{rmk}
\begin{rmk} \rm \label{r:pol}
{By choosing
 $$\aa=\hbar \sqrt{\frac{|c|}{2}}\left(\hbar\sqrt{\frac{|c|}{2}}\mp
 2\epsilon\right),$$
the operators $\hat{G}^\pm_\epsilon$ become polynomials in $\epsilon$ and  the relations
 (\ref{be}) reduce to
 $$\bb=\frac{c}{2}\left(\hbar\sqrt{\frac{|c|}{2}}\mp
  2\epsilon\right) 
  \left(\hbar(1-N)\sqrt{\frac{|c|}{2}}\pm
    2\epsilon\right), \qquad \ee=c_1.$$} 
\end{rmk}

\begin{rmk}
\rm For $c=0$ equations (\ref{beq}) and (\ref{eeq})
reduce to
$$
\bb^2=\frac{2}{\hbar^2}c_0\aa^2, \qquad \aa^2 c_1 = -\hbar^2\bb\ee.
$$
If $c_0>0$ we have  non trivial solutions for  $\aa\neq 0$: $\bb=\pm\frac{ \sqrt{2c_0}}{\hbar}\aa,$  $\ee=\mp\frac{c_1}{\hbar\sqrt{2c_0}}\aa$.
Moreover, the eigenvalue of $\hat{G}\psi_\lambda$ is  $\bar{\lambda}= \lambda-\hbar^2\bb/\aa=
\lambda\mp\hbar\sqrt{2c_0}$.
This ladder operator is then independent of $\lambda$ and does not seem useful to generalize the approach used in \cite{CKNd}. 
\end{rmk}

Since  the eigenvalue of $\hat{G}_\epsilon\psi_\lambda$ is of the form $(\epsilon+s)^2$ we have that
\begin{cor}
The operators (\ref{hGn})  map the eigenfunctions of $\hat{L}$ associated
with the eigenvalue $\lambda= \check{c}\epsilon^2$ to eigenfunctions
of $\hat{L}$ with eigenvalue $\bar{\lambda}=\check{c}(\epsilon + n s )^2$ for any integer $n>0$.
Moreover, for the choice of $\aa$ proposed in Remark \ref{r:pol}  
these operators  are polynomials in  $\epsilon$ whose coefficients can be easily determined explicitly.
\end{cor}

Proposition \ref{p:Ghat} provides a method for constructing ladder operators $\hat{G}_\epsilon$ starting from the spectrum of $\hat{L}$ by setting $\epsilon=\sqrt{|\lambda|}$
for all  eigenvalues $\lambda$ of $\hat{L}$ with suitable sign (recall that we need that $c\lambda \ge 0$).
In the following proposition, we show  how to avoid
the condition (\ref{cons}). By Remark \ref{r:L0}, we know that if a solution $G$ of (\ref{HE1})-(\ref{VE1}) is known
for $c\neq 0$ and arbitrary $c_0$, then the same $G$
satisfies the equations with the same $c$ and
$c_0=\dok c$ with $\dok$ given as in (\ref{dok}), 
if we add the constant term $\dok-c_0/c$ to $V$.
Indeed, we have

\begin{cor}
Let
$\hat{L}_0$ defined in (\ref{hL0}) be a Schr\"odinger operator such that equations
(\ref{HE1})-(\ref{VE1}) admit a solution $G$ for $c\neq 0$ and $c_0=0$. Then, the operators
$\hat{G}^\pm_\epsilon$ defined in (\ref{dhG}) are ladder operators for the eigenfunctions of
\begin{equation}\label{hLb}
\hat{L}_N=-\frac{\hbar^2}2\Delta+V_0- \frac{\hbar^2c(N-1)^2}{8}.
\end{equation}
\end{cor}
\begin{proof}
The potential $V$ of (\ref{hLb}) is of the form 
$V_0+\dok$. Hence, $\hat{L}_N$ verifies the hypotheses of Proposition \ref{p:Ghat}.
\end{proof}

Since eigenfunctions of $(\ref{hL0})$ and $(\ref{hLb})$ are the same and 
$$
\hat{L}_0 \psi_\lambda=\lambda\psi_\lambda\  \Rightarrow\ \hat{L}_N \psi_\lambda=(\lambda+\dok)\psi_\lambda
$$
that is the eigenvalues are shifted by the constant $\dok$, we try to construct ladder operators for $\hat{L}_0$:

\begin{prop}
Let $G$ be solution of (\ref{HE1}) and (\ref{VE1})
for $c \neq  0$ and for $c_0=0$.
By setting, for all $\lambda$ such that $c(\lambda+\dok)>0,$
\begin{equation}\label{ep0}
\epsilon=\sqrt{|\lambda+\dok|}=\sqrt{\left| \lambda -\frac{\hbar^2 c(N-1)^2}{8}\right|},
\qquad (i.e,\ \lambda=\check{c}\epsilon^2-\dok),
\end{equation} 
then, the operators $\hat G^\pm_\epsilon$ (\ref{dhG}) are  ladder operators for the eigenfunctions of the Hamiltonian operator (\ref{hL0}) i.e., for any $\psi_\lambda$ such that
$\hat{L}_0\psi_\lambda=\lambda\psi_\lambda$, we have
\begin{align}\label{L0G}
\hat{L}_0\hat{G}^-_\epsilon \psi_\lambda &=  
\left(\check{c}(\epsilon+s)^2-\dok\right)\hat{G}^-_\epsilon \psi_\lambda,\\
\hat{L}_0\hat{G}^+_\epsilon \psi_\lambda &= 
\left(\check{c}(\epsilon-s)^2-\dok\right)\hat{G}^+_\epsilon \psi_\lambda,  
\end{align}
with $s= \sqrt{\frac{|c|}2}\hbar$.
\end{prop}

\begin{proof}
For any $\psi_\lambda$ such that  $\hat{L}_0\psi_\lambda= \lambda \psi_\lambda$
we have $\hat{L}_N\psi_\lambda=\lambda_N \psi_\lambda$
with
$\lambda_N=\lambda+\dok$ and $\hat{L}_N$ as in (\ref{hLb}).
Hence, since $c(\lambda+\dok)>0$ means $c\lambda_N>0$,
 $\hat{L}_N$ satisfies the hypotheses of
Proposition \ref{p:Ghat}, then the operators $\hat{G}_\epsilon^\pm$ defined by (\ref{be})
for $\epsilon=\sqrt{|\lambda_N|}$
 verify (\ref{40}), where $\hat{L}$ is given by (\ref{hLb}). Thus, operator (\ref{hL0}) verifies (\ref{L0G}) and the relation between $\epsilon$ and $\lambda$ is given by (\ref{ep0}).
\end{proof}

\begin{cor}
The operators $(\hat{G}_\epsilon^\pm)^n$ defined in
(\ref{hGn}) are ladder operators for $\hat{L}_0$: that is,  
for any $\psi_\lambda$ such that
$\hat{L}_0\psi_\lambda=\lambda\psi_\lambda$, with
$c(\lambda+\dok)>0$, and
 for $\epsilon$ as in (\ref{ep0}), we have
\begin{align}\label{L0nG}
\hat{L}_0(\hat{G}^{-}_\epsilon)^n \psi_\lambda = 
\left(\check{c}(\epsilon+ns)^2-\dok\right) (\hat{G}^{-}_\epsilon)^n \psi_\lambda,
\\
\hat{L}_0(\hat{G}^{+}_\epsilon)^n \psi_\lambda = 
\left(\check{c}(\epsilon-ns)^2-\dok\right) (\hat{G}^{+}_\epsilon)^n \psi_\lambda.  
\end{align}
\end{cor}

\subsection{Quantum extended Hamiltonians $\hat{H}$ with Shift-ladder operators }\label{6.2}

As in \cite{KN0} we look for Shift-Ladder operators of the $M$-dependent radial operator (\ref{hHM}), that is we search for operators that modify the
value of $M$ as well as  the energy (they add a constant multiple of the identity to
$\hat{H}^M$). We use an ansatz on the form of $W(u)$ and of $\gamma(u)$, which allows us to exactly generalize the shift-ladder operators of the
the radial operator analysed in \cite{KN0} for $N=1$
to an arbitrary $N$. A study of a more general $W$
and $\gamma$ is not yet complete.
\begin{prop}\label{p:Ahat}
If the operators $\hat{H}^M$ (\ref{hHM}) verify $c\neq 0$,
 $C=0$ (that is $\gamma(u)=\frac{1}{cu}$) and
\begin{equation}\label{FF}
W(u)= \frac{\omega^2}2 c^2u^2 - \frac{\hbar^2(N-1)^2}{8u^2}=\frac{\omega^2}2c^2u^2-\frac{\dok}{cu^2}, \qquad \omega\in\mathbb{R},
\end{equation}
then, for all $M$ such that $cM>0,$ by setting
$$
\mu=\sqrt{|M|},
$$
the operators
(\ref{dhA}) satisfy
$$
 \hat{H}^{\bar{M}} \hat{A}^{\sigma,\tau}_\mu=\hat{A}^{\sigma,\tau}_\mu (\hat{H}^M -\sigma \hbar c\omega),
$$
with
\begin{equation}\label{sM}
\bar M=\check{c}\left(\mu-\tau s\right)^2=\check{c}\left(\mu-\tau\hbar \sqrt{\frac{|c|} 2}\right)^2.
\end{equation}
\end{prop}

\begin{proof}
We assume that the operators $\hat{A}$ are of the form
$$
\hat{A}= \partial_u + \left(a\gamma(u)+ b\frac{1}{\gamma(u)}\right), \qquad a,b\in \mathbb{R},
$$
then we find that a necessary condition for having
\begin{equation}\label{HAeq}
\hat H^{\bar M}\hat A=\hat A (\hat H^M- \delta),  \qquad \delta\in \mathbb{R}, 
\end{equation}
is $A=0$, that is $\gamma(u)=cu$, and	 
$$
\hat A= \partial_u+ \frac 1{\hbar^2 }\left[\delta u+\frac 1{2cu}(\hbar^2 c N-2(\bar{M}- M)) \right].
$$
Moreover, for $W(u)$ of the form $$W=\frac{\omega^2}{2}\gamma^{-2}+\beta \gamma^2,\qquad \omega,\beta\in \mathbb{R}$$  
the following conditions are sufficient in order to get the identity (\ref{HAeq}):
$$
\delta^2=\hbar^2 c^2\omega^2,
$$
(i.e., $\delta=\pm \hbar c \omega$)
and
$$
\bar{M}^2-\bar{M}(2M+\hbar^2c^2)+\frac{(2M-\hbar^2c^2)^2-\hbar^4c^4(N-1)^2-8h^2c^2\beta}{4}=0
$$
that is
$$
\bar{M}=M+\frac{\hbar^2c}{2}\pm \frac{\hbar}{2}\sqrt{\hbar^2c^2(N-1)^2+8\beta+8cM}.
$$
For $cM>0,$ by setting $\mu=\sqrt{|M|}$,
if
\begin{equation}\label{QC}
\beta=-\frac{\hbar ^2c^2(N-1)^2}{8}.
\end{equation}
then we have (\ref{sM}). Furthermore,
 $W$ becomes (\ref{FF}) and
we get the operators $\hat{A}_\mu^{\sigma,\tau}$ of the statement.
\end{proof}

\begin{rmk} \rm
Operators $\hat{A}_\mu^{\sigma,\tau}$ are a mixture of shift and ladder operators, as in \cite{CKNd}, since both the parameter $M$ and the energy of the Hamiltonian $\hat{H}^M$ are changed by them. Since we are not requiring that $\hat{H}^M$ can be written (up to a multiple of the identity) as a composition of $\hat{A}_\mu^{\sigma,\tau}$, we have more freedom in their form than in \cite{CKNd} (e.g. the $\hat{A}_\mu^{\sigma,\tau}$ are defined up to a constant factor).
\end{rmk}
\begin{rmk} \rm 
The value of $\beta$ in (\ref{QC}) coincides with the value of $c_0$ given by (\ref{cons}).
Moreover, the shift $s$ of $M^2$ coincides with the shift of $\epsilon^2$  produced by $\hat{G}_\epsilon$.
	\end{rmk}

\begin{cor}
If the radial operator $\hat{H}^M$ defined in (\ref{hHM}) is
\begin{equation}\label{hHok}
\hat{H}^M= -\frac{\hbar^2}2\left(\partial_u^2 + \frac{N}{u} \partial_u\right) 
+\frac{8M-\hbar^2c(N-1)^2}{8cu^2}+\frac{\omega^2}{2}c^2u^2, \qquad M=k^2\lambda.
\end{equation}
then, for all $M$ such that $cM>0$, by setting $\mu=\sqrt{|M|}$, $M=\check{c} \mu^2$,
for any positive integer $m$,  the operators $(\hat{A}^{\sigma,\tau}_\mu)^m$
(\ref{dAm})  verify
$$
(\hat{A}^{\sigma,\tau}_\mu)^m (\hat{H}^M-\sigma m\delta)=\hat{H}^{\check{c}(\mu-\tau ms)^2} (\hat{A}^{\sigma,\tau}_\mu)^m
$$
with $s=\sqrt{\frac{|c|}{2}}\hbar$, $\delta=\hbar c\omega$, $\sigma,\tau=\pm 1$.
\end{cor}

\begin{rmk} \rm
A direct attempt to reabsorb the constant $\dok$ in $M$, as we did for the ladder operators $\hat{G}$, does not seem to work. 
Indeed, we define 
the operator 
\begin{equation}\label{H0}
\hat{H}^0=-\frac{\hbar^2}{2}\left(\partial^2_{uu}+\frac{N}{u}\partial_u
\right)+ \frac{\omega^2}{2}c^2u^2.
\end{equation}
Let $\hat{H}_0^\nu$ be
\begin{equation}\label{hHmu}
\hat{H}_0^{\nu}= \hat{H}_0  
+\frac{\nu}{cu^2},
\end{equation}
then, we have the relations
$$\hat{H}^M=\hat{H}_0^{{M+\dok}},
\qquad 
\hat{H}_0^\nu=\hat{H}^{{\nu-\dok}},
$$
with $\dok$ as in (\ref{dok}).
We look for shift-ladder operator of $\hat{H}_0^\nu$.
By 
setting $$\mu=\sqrt{|\nu-\dok|},$$
and considering the operators defined in Proposition \ref{p:Ahat}, we get
$$
\hat{A}^{\sigma,\tau}_\mu\hat{H}_0^\nu=(\hat{H}_0^{\bar{\nu}} +\sigma \delta) \hat{A}^{\sigma,\tau}_\mu,$$
where $\bar{\nu}=\check{c}(\sqrt{|\nu-\dok|}-\tau s)^2 +\dok$, $s=\sqrt{\frac{|c|}{2}}\hbar$, $\delta=\hbar c\omega$, and $\sigma,\tau=\pm 1$.
The form of $\bar{\nu}$ allows to apply also the operator $(\hat{A}^{\sigma,\tau}_\mu)^m$, which for $\mu$ as above gives
\begin{equation}\label{coso}
(\hat{A}^{\sigma,\tau}_\mu)^m \hat{H}^\nu_0 =(\hat{H}^{\bar{\nu}}_0 +m\sigma\delta
)(\hat{A}^{\sigma,\tau}_\mu)^m,
\end{equation}with $\bar{\nu}=\check{c}(\mu-m\tau s)^2 +\dok$.
For any function $\phi(u)$ and for any eigenfunction
$\psi_\lambda$ of $\hat{L}_0$ (\ref{hL0}) the operator
$$
\hat{H}_1= \hat{H}_0-\frac{k^2}{cu^2}\hat{L}_0
$$
acts on $(\hat{G}^{-}_\epsilon)^{2n}(\hat{A}_\mu^{\sigma,\tau})^{2m}(\phi \psi_\lambda)$  (where $\epsilon$ is given by (\ref{ep0}), and
 $\mu$ has to be determined) as
\begin{equation}\label{pezzo}
\hat{H}_1((\hat{G}^{-}_\epsilon)^{2n} \circ(\hat{A}_\mu^{\sigma,\tau})^{2m}(\phi \psi_\lambda))=
\hat{H}_1[(\hat{G}^{-}_\epsilon)^{2n}(\psi_\lambda) (\hat{A}_\mu^{\sigma,\tau})^{2m}(\phi)].
\end{equation}
Since $(\hat{G}^{-}_\epsilon)^{2n}(\psi_\lambda)$ is an eigenfunction of $\hat{L}_0$ associated with the eigenvalue $\bar{\lambda}=\check{c}(\epsilon+2ns)^2-\dok$, we get that (\ref{pezzo}) equals
$$
 \hat{H}_0^{\bar{\nu}}[(\hat{A}_	\mu^{\sigma,\tau})^{2m}(\phi)] (\hat{G}^{-}_\epsilon)^{2n} (\psi_\lambda)
= \hat{H}^{\bar{\nu}-\dok}[(\hat{A}_	\mu^{\sigma,\tau})^{2m}(\phi)](\hat{G}^{-}_\epsilon)^{2n}(\psi_\lambda)
$$
with $\bar{\nu}
=\check{c}{(k\epsilon+2ms)^2-k^2\dok}$.
But it  does not seem possible to go further and apply (\ref{coso}) for any $\mu$ if $\dok \neq 0$.
\end{rmk}

\subsection{Shift operators for the eigenfunctions of $\hat H$} \label{6.3}
Let $\phi^M_E(u)$ be an eigenfunction of (\ref{hHok})  with eigenvalue $E$:
$
\hat{H}^M(\phi^M_E)=E\phi^M_E.
$
We look for shift operators for these eigenfunctions.
\begin{prop}\label{p:Dhat}
The operators
(\ref{Dhat}) satisfy
\begin{equation}\label{Deq1n}
\hat H^M \hat D^\pm_E \phi^M_E=(E\pm 2\hbar c \omega)\hat D^\pm_E \phi^M_E,
\end{equation}
i.e., they map eigenfunctions of $\hat{H}^M$ to eigenfunction of the same operator with
a different eigenvalue:
$$\hat{D}^+_E: \phi^M_E \mapsto \phi^M_{E+2\hbar c\omega}, \qquad
\hat{D}^-_E: \phi^M_E \mapsto \phi^M_{E-2\hbar c\omega},
$$
\end{prop}
\begin{proof}
By using the ansatz 
$$
\hat{D}_E\phi=a_0 cu\partial_u\phi+(a_1+a_2c^2u^2)\phi
$$
($a_i\in\mathbb{R},$ $i=0,1,2$)
and imposing the condition that $\phi$ has to be
an eigenfunction of $\hat{H}^M$,
we compute $\hat{H}^M(\hat{D}_E(\phi^M_E))$
and we determine the values of $a_1=\mp \frac{c}{\hbar}\omega a_0$, $a_2=\left(c\frac{N+1}{2}\pm \frac{E}{\hbar \omega } \right)a_0$
and the shift of the eigenvalue  (\ref{Deq1n}). For $a_0=\hbar/(\sqrt{2}c)$ we obtain
(\ref{Dhat}).
\end{proof}

Up to a constant factor, the relations (\ref{Deq1n}) coincide with the analogous ones in \cite{CKNd} for $\hbar^2=2$. 

The expressions of the operators $\hat{D}^\pm_E$ and
the shift of the eigenvalue do not depend on $M$.
Hence, the same operators are ladder operators for
the eigenfunctions of the operator (\ref{hHmu}).

The operators (\ref{dDm})
verify the shift property: if $\hat H^M  \phi^M_E=E\phi^M_E$, then
\begin{equation}\label{hHD}
\hat H^M (\hat D^{+ }_E)^m \phi^M_E=(E+ 2m\delta)(\hat D^{+ }_E)^m \phi^M_E, \quad
\hat H^M (\hat D^{- }_E)^m \phi^M_E=(E- 2m\delta)(\hat D^{- }_E)^m \phi^M_E,
\end{equation}
with $\delta=\hbar c \omega$.

\subsection{Proof of Theorem \ref{teoq} }\label{6.4}

\begin{proof1}{Theorem \ref{teoq}}
 Let $\hat{L}_0$ be the Schr\"odinger operator defined in (\ref{hL0})
such that  (\ref{HE1}) and (\ref{VE1}) admit a non trivial solution $G$ with $c\neq 0$ and $c_0=0$. 
Let 
 $\hat{H}^0$ be the operator defined in (\ref{H0}),
then we can write  the operator (\ref{Hokk}) as
\begin{equation}\label{Hok}
\hat{H}=\hat{H}^0+\frac{k^2}{cu^2}\hat{L}_0 +(k^2+1)\frac{\dok}{cu^2},
\end{equation}
with $k\in\mathbb{Q}^+$ and $\dok$ given by (\ref{dok}).
Moreover, we can write the $M$-dependent operator (\ref{hHok}) as
$$
\hat{H}^M=\hat{H}^0+ \frac{M+\dok}{cu^2}. 
$$

Let $\psi_\lambda$ be an eigenfunction of $\hat{L}_0$ associated with an eigenvalue $\lambda$ such that $c(\lambda+\dok)>0$.
 Then, the function $\phi(u)\psi_\lambda$ is an eigenfunction of $\ref{Hok}$ with eigenvalue $E$, if and only if
 $\phi_E(u)=\phi_E^{M}$ with
 $\hat{H}^{M}(\phi_E^{M})=E\phi_E^{M}$ for $$M=k^2(\lambda+\dok).$$ Obviously, $cM>0$ if and
 only if $c(\lambda+\dok)>0$. 
We set $\epsilon$ as in (\ref{ep0}), so that
$$\sqrt{|M|}=k\epsilon,$$ or, equivalently,$$ M=\check{c}k^2\epsilon^2.$$
According to Definition \ref{d:ws}, we prove that 
 $$\hat{H}(\hat{X}^k_{\epsilon, E} (\phi_E^{M}\psi_\lambda))= E\hat{X}^k_{\epsilon, E} (\phi_E^{M}\psi_\lambda),$$ for all multiplicatively separated eigenfunctions $\phi_E^{M}\psi_\lambda$ of $\hat{H}$ such that $c(\lambda+\dok)>0$.
We have the splitting $$\hat{X}^k_{\epsilon, E}(\phi_E^{M}\psi_\lambda)= (\hat{A}^{1,1}_{k\epsilon})^{2m} \circ (\hat{D}^{+}_E)^m(\phi_E^{M})\cdot (\hat{G}^{+}_\epsilon)^{2n}( \psi_\lambda).$$
Hence, by applying
Propositions \ref{p:Ghat},
\ref{p:Ahat} and \ref{p:Dhat},  
we have

\begin{align*}
& \hat{H}(\hat{X}^k_{\epsilon, E}  (\phi_E^{M}\psi_\lambda))=
\\
=&
\hat{H}^0
(\hat{A}^{1,1}_{k\epsilon})^{2m} 
 (\hat{D}^{+}_E)^m(\phi_E^M)\cdot (\hat{G}^{+}_\epsilon)^{2n}( \psi_\lambda) + \\
  &+ (\hat{A}^{1,1}_{k\epsilon})^{2m} 
  (\hat{D} ^{+}_E)^m(\phi_E^{M})\cdot \frac{k^2}{cu^2}
\hat{L}_0(\hat{G}^{+}_\epsilon)^{2n}( \psi_\lambda)
+
\\
 &+ (\hat{A}^{1,1}_{k\epsilon})^{2m} 
 (\hat{D} ^{+}_E)^m(\phi_E^{M})\cdot
\frac{(k^2+1)\dok}{cu^2}(\hat{G}^{+}_\epsilon)^{2n}( \psi_\lambda)=
\\
=&
\hat{H}^0(\hat{A}^{1,1}_{k\epsilon})^{2m} 
 (\hat{D} ^{+}_E)^m(\phi_E^{k\epsilon})\cdot (\hat{G}^{+}_\epsilon)^{2n}( \psi_\lambda) + \\
  &+ (\hat{A}^{1,1}_{k\epsilon})^{2m} 
  (\hat{D} ^{+}_E)^m(\phi_E^{M})\cdot
\frac{\check{c}k^2(\epsilon+2ns)^2+\dok}{cu^2}
(\hat{G}^{+}_\epsilon)^{2n}( \psi_\lambda)=
\\
=&
[\hat{H}^{\check{c}k^2(\epsilon-2ns)^2} (\hat{A}^{1,1}_{k\epsilon})^{2m} 
 (\hat{D} ^{+}_E)^m(\phi_E^{M})]\cdot (\hat{G}^{+}_\epsilon)^{2n}( \psi_\lambda)=
\\
=&
[\hat{H}^{\check{c}(k\epsilon-2ms)^2} (\hat{A}^{1,1}_{k\epsilon})^{2m} 
 (\hat{D} ^{+}_E)^m(\phi_E^{M})] \cdot (\hat{G}^{+}_\epsilon)^{2n}( \psi_\lambda)=
\\
=&
[(\hat{A}^{1,1}_{k\epsilon})^{2m} \circ (\hat{H}^{M}-2m\delta)\circ (\hat{D} ^{+}_E)^m(\phi_E^{M})]\cdot
(\hat{G}^{+}_\epsilon)^{2n}( \psi_\lambda)=
\\
=&
(\hat{A}^{1,1}_{k\epsilon})^{2m}[(E+2m\delta)
(\hat{D} ^{+}_E)^m(\phi_E^{M})-
2m \delta (\hat{D} ^{+}_E)^m(\phi_E^{M})]
\cdot (\hat{G}^{+}_\epsilon)^{2n}( \psi_\lambda)=
\\
=&
E(\hat{A}^{1,1}_{k\epsilon})^{2m}\circ (\hat{D} ^{+}_E)^m(\phi_E^{M})\cdot (\hat{G}^{+}_\epsilon)^{2n}( \psi_\lambda)=
\\
=&
E\hat{X}(\phi_E^{M}\psi_\lambda
).
\end{align*}
This proves the statement.
\end{proof1}

%
%

\section{Conclusions}
We analysed the Kuru-Negro ansatz about the existence of ladder functions in the context of extended Hamiltonian systems. We showed that the ansatz fits naturally in this theory, leading always to a factorized form of the  characteristic first integrals of extended systems. Moreover, we proved that for a large subclass of extended  Hamiltonians the ansatz is not restrictive. 

In the quantum context, 
instead of a deduction of the classical ladder functions and factorized first integrals from the
quantum factorized symmetry operators constructed through shift and ladder operators, by a sort of semiclassical limit procedure, we worked in the opposite direction: from the properties of the classical Hamiltonian $L$ we constructed a quantum operator $\hat{H}$ admitting a factorized symmetry operator. Indeed, 
we have considered the quantum Hamiltonian associated with a classical Hamiltonian $L$ admitting a ladder function.
We have explicitly constructed an extended quantum 
Hamiltonian and four factorized symmetry operators, 
by using shift and ladder operators, as in the Kuru-Negro theory. The quantum extension is a Schr\"odinger operator, whose classical counterpart is similar
to a particular case (case $c\neq 0$, $C=0$) of an extension of $L$. 
However, the two Hamiltonians coincide only if $L$ has one degree of freedom.  

These results show the intimate connection between the theory of extended Hamiltonians and the procedure of factorization in shift-ladder operators. The connection reveals itself in essentially two properties: the possibility of a partial separation of variables $(u,q^i)$ and the existence of solutions of the equation $X_L^2G=-2(cL+c_0)G$. Both the properties are consequences of the structure of warped manifold: the first one, of the configuration manifold of $H$, the second one, of the configuration manifold of $L$ \cite{Ta}. The relations between factorization and warped manifolds owe to be studied in greater detail. The opportunity of several further studies emerges clearly: 

(i) the construction of factorized symmetry operators should be carried on also  for the general extended metric tensor (case $C\neq 0$, $c\neq 0$); 

(ii) the  shift and ladder operators are here built on the model of \cite{CKNd}, but different, more general, forms of them should be investigated: in particular, a quantization based on the factorization results  of Theorem \ref{GN} should be studied; 

(iii) it is an open problem if the ansatz is actually a restrictive condition for the more general $G(q^i,p_j)$.

\end{document}